\documentclass[aps,pre,preprint]{revtex4-1}
\usepackage{enumerate}
\usepackage[utf8]{inputenc}
\usepackage{amsthm}
\usepackage{amssymb}
\usepackage{amsmath}

\newtheoremstyle{definition}
{\topsep}
{\topsep}
{\em}
{}
{\bf}
{:}
{.5em}
{\thmname{#1}\thmnumber{ #2}\thmnote{ (#3)}}
\theoremstyle{definition}
\newtheorem{definition}{Definition}

\newtheoremstyle{proposition}
{\topsep}
{\topsep}
{\em}
{}
{\bf}
{:}
{.5em}
{\thmname{#1}\thmnumber{ #2}\thmnote{ (#3)}}
\theoremstyle{proposition}
\newtheorem{proposition}{Proposition}

\newtheoremstyle{lemma}
{\topsep}
{\topsep}
{\em}
{}
{\bf}
{:}
{.5em}
{\thmname{#1}\thmnumber{ #2}\thmnote{ (#3)}}
\theoremstyle{lemma}
\newtheorem{lemma}{Lemma}

\newtheoremstyle{corollary}
{\topsep}
{\topsep}
{\em}
{}
{\bf}
{:}
{.5em}
{\thmname{#1}\thmnumber{ #2}\thmnote{ (#3)}}
\theoremstyle{corollary}
\newtheorem{corollary}{Corollary}

\begin{document}
	
	
	\title{Wealth concentration in systems with unbiased binary exchanges}

\author{Ben-Hur Francisco Cardoso}
\email{ben-hur.cardoso@ufrgs.br}
\affiliation{Instituto de Física, Universidade Federal do Rio Grande
	do Sul, Porto Alegre, RS, Brazil}
\author{Sebastián Gonçalves}
\email{sgonc@if.ufrgs.br}
\affiliation{Instituto de Física, Universidade Federal do Rio Grande
	do Sul, Porto Alegre, RS, Brazil}
\author{José Roberto Iglesias}
\email{iglesias@if.ufrgs.br}
\affiliation{Instituto de Física, Universidade Federal do Rio Grande
	do Sul, Porto Alegre, RS, Brazil}
\affiliation{Instituto Nacional de Ciência e Tecnologia de Sistemas
	Complexos, INCT-SC, CBPF, Rio de Janeiro, RJ, Brazil}

	\date{\today}
	
	\begin{abstract}
		\noindent
                Aiming to describe the wealth distribution evolution, several models consider 
an ensemble of interacting economic agents that exchange wealth in binary fashion. Intriguingly, 
models that consider an unbiased market, that gives to each agent the same chances to win in the 
game, 
are always out of equilibrium until the perfect inequality of the final state is attained.
Here we present a rigorous analytical demonstration that any system driven by unbiased binary 
exchanges are doomed to drive the system to perfect inequality and zero mobility.
	\end{abstract}
	
	\maketitle

\section{Introduction}
Statistical physics, in particular the kinetic theory of gases, provides a useful framework for describing the complexity of microscopic market interactions~\cite{slanina2013essentials}. Like a physical system composed of many particles that exchange their energy in binary conservative collisions, {\em the kinetic exchange models}~\cite{ispolatov1998wealth, boghosian2017oligarchy, patriarca2010basic, yakovenko2009colloquium} consider a set of interacting agents where, sequentially, two randomly chosen agents exchange a conserved quantity called ``wealth".
Following this approach, given a system with $N$ agents, each agent $i$ is characterized by the wealth $x_i \geq 0$. So in an exchange between agents $i$ and $j$, we have:
\begin{equation}
x_{i}^* = x_{i} + \Delta_{i} \>\>\> \text{and} \>\>\> x_{j}^* = x_{j} + \Delta_{j},
\end{equation}
where $x_{i(j)}^*$ is the wealth after the exchange and $\Delta_{i(j)}$ is the stochastic gain of agent $i(j)$. Since the kinetic exchange models work as zero-sum games, we have that $\Delta_{j} = -\Delta_{i}$.

The first model within this general framework is the so-called {\it loser} rule~\cite{hayes2002, angle1986surplus}, where the  wealth exchanged is 
\begin{equation}
\Delta_{i} = \epsilon \lambda x_{j} - (1 - \epsilon) \lambda x_{i},\>\>\epsilon\in\{0,1\},\>\>\mathbb{E}[\epsilon] = \frac{1}{2},
\end{equation}
and $0 \leq \lambda \leq 1$ can be a random or constant number. The model was criticized by {\em Lux}~\cite{lux2005emergent}, noting that within this mechanism $\mathbb{E}[\Delta_i] \propto (x_j - x_i)$; so it favors, on average, the poorer agents.  Why then would a rational agent participate in this type of transaction with a partner less wealthy than him/her? The same criticism could be applied for the variations of the {\it loser rule}~\cite{cardoso2020wealth,chatterjee2004pareto,chakraborti2000statistical}.
It is important to note that the criticism is not about the existence of a poor-biased wealth dynamics, but the assumption that such a mechanism emerges from free market exchanges. 
To stabilize the wealth distribution, some kind of poor-favoring mechanism is necessary, but, as we will see, in the form of a regulation policy.

In contrast, other models overcome the criticisms by proposing rules where the expected wealth gain is the same, regardless of whether the agent is rich or poor.
Mathematically, a zero-sum exchange process is unbiased~\cite{bouleau2017impact} if $\mathbb{E}[\Delta_{i}] = \mathbb{E}[\Delta_{j}] = 0$. The best-known unbiased model is the {\it Yard-Sale}~\cite{hayes2002} 
\begin{equation}
\Delta_{i} = \eta\lambda\min(x_{i},x_{j}),\>\>\eta\in\{-1,1\},\>\>\mathbb{E}[\eta] = 0,
\end{equation}
where $0\leq\lambda\leq 1$ can be a constant or a random number.
In addition, it was proposed recently a modified version of the {\it loser} rule~\cite{bouleau2017impact}, stating that
\begin{equation}
\Delta_{i} = \epsilon\lambda x_{j} -  (1 - \epsilon)\lambda x_{i},\>\>\epsilon\in\{0,1\},\>\>\mathbb{E}[\epsilon] = \frac{x_{i}}{x_{i} + x_{j}},
\end{equation}
where $0 \leq \lambda \leq 1$ can be a random or constant number. Finally, a similar model was proposed by Iglesias and Almeida~\cite{iglesias2012entropy}, where
\begin{equation}
\Delta_{i} = \eta\frac{x_{i}x_{j}}{x_{i} + x_{j}},\>\>\eta\in\{-1,1\},\>\>\mathbb{E}[\eta] = 0.
\end{equation}

Although all these rules are supposed to be unbiased, the system always converges to {\it condensation}~\cite{moukarzel2007wealth, iglesias2012entropy,cardoso2020wealth, boghosian2017oligarchy,bouleau2017impact}, where one or a few agents concentrate all the wealth and there are no more exchanges. To avoid condensation, it is necessary to include some bias toward the poorer in the models, such as taxation~\cite{lima2020nonlinear,iglesias2020inequality, boghosian2017oligarchy,bouleau2017impact, li2019affine}, redistribution~\cite{iglesias2020inequality}, or a factor that increases the probability of the poorest agent to win in an exchange~\cite{moukarzel2007wealth,cardoso2020wealth,iglesias2012entropy}.

The  rest of the paper is organized as follows.
In section~\ref{sec:master} we will define the economic metrics of inequality and mobility and describe the evolution of unbiased kinetic exchange models in the thermodynamic limit ($N\rightarrow \infty$), following a Boltzmann-like Master Equation~\cite{ispolatov1998wealth, yakovenko2009colloquium, boghosian2014kinetics}.
In section~\ref{sec:prop} we will proof that condensation is the asymptotic state of any unbiased kinetic exchange model, in the thermodynamic limit. In other words, we will demonstrate that wealth concentration and mobility reduction is a natural consequence of free exchanges in an unregulated market, with rules that are in principle unbiased or ``fair''. Section IV is devoted to the conclusions.

\section{Master Equation}\label{sec:master}

\begin{definition}
	Let be a system represented by the probability density function $f(x,t)$, where $f(x,t)dx$ is the fraction of agents with wealth within $dx$ of $x$ at time $t$. This probability density function must obey the following properties:
	\begin{enumerate}[(i)]
		\item Negative wealth is not allowed, then, for all $t$ we have that
		\begin{equation}
		x < 0 \Rightarrow f(x,t) = 0.
		\end{equation}
		\item For all $t$, the density distribution function must be normalized
		\begin{equation}
		\int_0^{\infty}dx\>f(x,t) = 1.
		\end{equation}
		\item The wealth is conserved. Thus, for all $t$, the first moment of the density distribution function must be constant
		\begin{equation}
		\int_0^{\infty}dx\>x\>f(x,t) = \langle x \rangle.
		\end{equation}
	\end{enumerate}
\end{definition}

\begin{definition}
	If the dynamics of $f(x,t)$ is driven by binary exchanges, then we can represent this dynamics with the following Master Equation~\cite{ispolatov1998wealth, boghosian2017oligarchy, bustos2016wealth, yakovenko2009colloquium,bassetti2010explicit} 
	\begin{eqnarray}\label{eq:master}
	\nonumber
	\frac{\partial f(x,t)}{\partial t} = \int_{0}^{\infty} dx' \int_{-x}^{x'} d\Delta \bigg\{ \omega_{[x+\Delta,x'-\Delta]\rightarrow[x,x']}f(x+\Delta,t) f(x'-\Delta,t) -\\
	-\omega_{[x,x']\rightarrow[x+\Delta,x'-\Delta]}f(x,t) f(x',t)\bigg\}, 
	\end{eqnarray}
	where  $\omega_{[x,x']\rightarrow[x+\Delta,x'-\Delta]}\>d\Delta\>f(x,t)\>dx \>f(x',t)\>dx'$ is the rate of transferring wealth within $d\Delta$ of $\Delta$ from an agent with wealth within $dx$ of $x$ to an agent with wealth within $dx'$ of $x'$~\cite{yakovenko2009colloquium,pitaevskii2012physical}.
	Considering that all possible values of $\Delta$ must be in the interval $-x \leq \Delta \leq x'$, to keep wealth non-negative, along with the fact that wealth is neither created nor destroyed, they impose the following normalization on transfer rates:
	\begin{equation}
	\int_{-x}^{x'} d\Delta \>\omega_{[x,x']\rightarrow[x+\Delta,x'-\Delta]} = 1.
	\label{eq:normal}
	\end{equation}
\end{definition}

\begin{definition}
Let be two random selected agents with wealth $x$ and $x'$. An exchange process is said unbiased if the expected gain is zero for both agents~\cite{bouleau2017impact,cardoso2020wealth}, that is,
\begin{equation}
\int_{-x}^{x'} d\Delta \>\Delta \>\omega_{[x,x']\rightarrow[x+\Delta,x'-\Delta]} = 0.
\label{eq:unbiased}
\end{equation}
\end{definition}

\subsection{Wealth Inequality}

\begin{definition}
	Let be a probability distribution function $f(x,t)$ described by the Definition 1. 	The Gini index is a measure of the inequality  that is defined as~\cite{sen1997economic}:
	\begin{equation}\label{eq:gini}
	G(t) = \frac{1}{2\langle x \rangle}\int_{0}^{\infty} dx \int_{0}^{\infty} dx_1\> | x-x_1 | \> f(x, t) f (x_1, t).
	\end{equation}
The Gini index varies between $0$, which corresponds to perfect equality (i.e. everyone has the same wealth), and $1$, that corresponds to maximum inequality (the absolute oligarchy~\cite{boghosian2017oligarchy}).
\end{definition}

It is important to note that its difficult to translate the notion of maximum inequality for a finite number of agents (i.e. one person has all the wealth, while everyone else has zero wealth) to a probability density function, since it has an ``infinite'' number of agents. That problem was studied by Boghosian et al.~\cite{boghosian2014kinetics, boghosian2015h, boghosian2017oligarchy}, with the following definition:

\begin{definition}
	The absolute oligarchy density distribution function is defined as
	\begin{equation}
	f_c (x) = \delta (x) + \lim_{N \rightarrow \infty} \frac{\delta (x - \langle x \rangle N)}{N},
	\label{eq:oli}
	\end{equation}
	that is, all wealth is in hands of an infinitesimal part of the system. Also, it is the unique probability density function, in terms of Definition 1, that satisfies the following properties ~\cite{boghosian2014kinetics, boghosian2015h, boghosian2017oligarchy}:
    \begin{enumerate}[(i)]
		\item $f_c(x) > 0 \iff x = 0$.
		\item $\int_{0}^{\infty}dx\>x\>f_c(x) = \langle x \rangle$.
		\item for all $k > 1$, $\int_{0}^{\infty}dx\>x^k\>f_c(x)$ diverges.
	\end{enumerate}

In fact, the property $(i)$ of the previous definition says that, with the absolute oligarchy probability density function, the probability of having agents with positive wealth is zero, that is,
\begin{equation}
\mathbb{P}(x > 0) = 0.
\end{equation}
\end{definition}

This does not mean, however, that there are no agents with $x> 0$. This apparent paradox stems from the nature of probability measures in infinite sets, as is the case with the thermodynamic limit ($ N \rightarrow \infty $). For the probability to be zero, it is sufficient that there be a finite~\cite{schilling2017measures} --not necessarily null-- number of agents with positive wealth.
In the example of absolute oligarchy, the number of agents that concentrates all wealth is $ 1 $, regardless of the total number of agents $ N $. At the thermodynamic limit, the total set of agents becomes infinite. Thus, since the set of agents with positive wealth is finite ($ 1 $ agent), the probability of having agents with positive wealth becomes zero.

\subsection{Mobility}
Keeping the constraint of wealth conservation, the flux of wealth from one agent to another, that we call ``mobility'',  is only possible through exchanges. Due to this fact, we define:

\begin{definition}
	The mobility of agents with wealth $x$ at time $t$ is defined by	\begin{equation}\label{eq:mob}
	l(x,t) = \int_0^{\infty}dx'\int_{-x}^{x'} d\Delta \>|\Delta|\> \omega_{[x,x']\rightarrow[x+\Delta,x'-\Delta]}\> f(x',t).
	\end{equation}
\end{definition}

\begin{definition}
	The wealth $x$ is said to be in an absorbing state if $l(x,t) = 0$ whenever $f(x,t) > 0$,  i.e, once the agent is in this state, she/he can no longer get out of this situation.
\end{definition}

\begin{definition}
	
	The system's mobility is the fraction of average wealth exchanged per unit time, called the {\it Liquidity} of the system~\cite{iglesias2012entropy,cardoso2020wealth}, defined as
	\begin{equation}\label{eq:liquidez}
	L(t) = \frac{1}{2\langle x \rangle}\int_{0}^{\infty} dx\>l(x,t) f(x,t),
	\end{equation}
	where the factor $1/2$ cancel the double counting. Liquidity varies between $0$, which corresponds to an economy  without exchanges, and $1$, that corresponds to full exchange economy (i.e. all the wealth are exchanged per unit time).
\end{definition}

\section{Propositions}\label{sec:prop}

\begin{proposition}
A system of unbiased binary exchanges has $x = 0$ as an absorbing state.
\end{proposition}

The proposition 1 has already been observed, considering finite size systems, in all the unbiased models reviewed in the Introduction~\cite{moukarzel2007wealth, iglesias2012entropy, cardoso2020wealth, boghosian2017oligarchy, bouleau2017impact}.

\begin{proposition}
In a system of unbiased binary exchanges, the Gini index is monotonically increasing:
	\begin{equation}
	\frac{dG(t)}{dt} \geq 0.
	\end{equation}
\end{proposition}

The proposition 2 has already been observed, considering finite size systems, in all the unbiased models reviewed in the Introduction~\cite{moukarzel2007wealth, iglesias2012entropy, cardoso2020wealth, boghosian2017oligarchy, bouleau2017impact}. In the thermodynamic limit, with a similar methodology introduced in this section, the particular case of the {\it Yard Sale} rule was proved by {\it Boghosian et al.}~\cite{boghosian2015h}.

Also, as a consequence of proposition 2, the Gini index can be considered a Lyapunov functional of the systems driven by binary unbiased wealth exchanges~\cite{boghosian2015h}.

\begin{proposition}
In a system of binary unbiased exchanges, if $x = 0$ is the unique absorbing state of the system, then
	\begin{enumerate}[(i)]
		\item The stationary probability density function is the absolute oligarchy one (Eq.~\ref{eq:oli})
		\begin{equation}
		\lim_{t\rightarrow\infty}f(x,t) = f_c (x)
		\end{equation}
		\item The stationary inequality, so, is the highest one
		\begin{equation}
		\lim_{t\rightarrow\infty}G(t) = 1
		\end{equation}
		\item The stationary liquidity is the lowest one
		\begin{equation}
		\lim_{t\rightarrow\infty}L(t) = 0
		\end{equation}
	\end{enumerate}
\end{proposition}

The proposition 3 has already been observed, considering finite size systems, in all the unbiased models reviewed in the Introduction~\cite{moukarzel2007wealth, iglesias2012entropy, cardoso2020wealth, boghosian2017oligarchy, bouleau2017impact}.

\subsection{Proof of propositions}

Before the proof of the propositions, we prove a corollary of a lemma that will be necessary 

\begin{lemma} Let be $g : [a, b] \rightarrow \mathbb{R}$ a function such that $\int_a^bdzg(z) = 1$. Then~\cite{zabandan2013several}
	\begin{equation}
	\int_a^bdz|z|g(z) \leq \frac{b|a| - a|b|}{b-a} + \frac{|b|-|a|}{b-a}\int_a^bdzzg(z).
	\end{equation}
\end{lemma}

\begin{corollary}\label{cor} In a system of unbiased binary exchanges, for all $t$ and $x>0$, we have that
	\begin{equation}
	l(x,t) \leq 2\langle x \rangle.
	\end{equation}
\end{corollary}
\begin{proof}[{\bf proof of corollary 1}]
	By the definition of mobility, the Lemma 1 implies that
	\begin{eqnarray}
	\nonumber
	l(x,t) = \int_0^{\infty}dx'\int_{-x}^{x'} d\Delta \>|\Delta|\> \omega_{[x,x']\rightarrow[x+\Delta,x'-\Delta]} f(x',t) \leq \\
	\leq\int_0^{\infty}dx' f(x',t)\bigg[\frac{2xx'}{x+x'} + \frac{x-x'}{x+x'} \int_{-x}^{x'} d\Delta \>\Delta\> \omega_{[x,x']\rightarrow[x+\Delta,x'-\Delta]} \bigg].
	\end{eqnarray}
	
	By the condiciton of unbiased exchanges (Eq.~\ref{eq:unbiased}), the second term in the sum is zero. Then, we get
	\begin{equation}
	l(x,t) \leq 2\int_0^{\infty}dx' f(x',t)\>\frac{xx'}{x+x'}.
	\end{equation}
	
	now, since $x \leq x + x'$, the finally poof the corollary
	\begin{equation}
	l(x,t) \leq 2\int_0^{\infty}dx' f(x',t)\>x' = 2\langle x\rangle.
	\end{equation}
\end{proof}

\begin{proof}[{\bf proof of proposition 1}]
	Let assume that $x = 0$ e $x' > 0$. So, the Eq.~\ref{eq:unbiased} implies that
	\begin{equation}
	\int_{0}^{x'} d\Delta \>\Delta \>\omega_{[0,x']\rightarrow[\Delta,x'-\Delta]} = 0.
	\end{equation}
	Since in this particular case the integrand is always non-negative, the integral only becomes zero if
	\begin{equation}\Delta \>\omega_{[0,x']\rightarrow[\Delta,x'-\Delta]} = 0\>\>\>,\>\>\>\forall \Delta \in [0,x'].
	\end{equation}
	
	So, regarding the normalization (Eq.~\ref{eq:normal}), this condition is only satisfied if
	\begin{equation}\label{eq:ab}
	\omega_{[0,x']\rightarrow[\Delta,x'-\Delta]} = \delta(\Delta).
	\end{equation}
	
	Now, using the Eq.~\ref{eq:mob}, we finally have
	\begin{equation}
	l(0,t) = \int_0^{\infty}dx'\int_{0}^{x'} d\Delta \>|\Delta|\>  \delta(\Delta) f(x',t) = 0,
	\end{equation}
	that is, $x=0$ is an absorbing state.
\end{proof}

\begin{proof}[{\bf proof of proposition 2}]
As written below, the Gini Index at time $t$ is given by
	\begin{equation}
	G(t) = \frac{1}{2\langle x\rangle}\int_{0}^{\infty} dx \int_{0}^{\infty} dx_1\> | x-x_1| \> f(x, t) f (x_1, t).
	\end{equation}
	So, its time evolution is given by 
	\begin{eqnarray}\label{eq:gdt}
	\frac{dG(t)}{dt} = 
	\frac{1}{2\langle x\rangle}\int_{0}^{\infty} dx \int_{0}^{\infty} dx_1\> | x-x_1| \>\bigg[\frac{\partial f(x,t)}{\partial t}f(x_1,t) + \frac{\partial f(x_1,t)}{\partial t}f(x,t)\bigg].
	\end{eqnarray}
	Since the integration is over the entire range of $x$ and $x_1$ and they play a symmetric role, both terms inside the brackets contribute in the same way. Then, Eq.~\ref{eq:gdt} simplifies to
	\begin{equation}
	\frac{dG(t)}{dt}= \frac{1}{\langle x \rangle}\int_{0}^{\infty} dx \int_{0}^{\infty} dx_1\> | x-x_1| \>f(x_1,t)\>\frac{\partial f(x,t)}{\partial t}.
	\label{eq:gdt2}
	\end{equation}
	Defining 
    \begin{equation}\label{eq:phi}
    \phi(x,t) \equiv  \int_{0}^{\infty} d x_1\> | x-x_1| \>f(x_1,t),
    \end{equation}
    we rewrite Eq.~\ref{eq:gdt2} as 
    \begin{equation}
    \frac{dG(t)}{dt}= \frac{1}{\langle x \rangle}\int_{0}^{\infty} dx  \>\phi(x,t) \>\frac{\partial f(x,t)}{\partial t}.
    \label{eq:gdt3}
    \end{equation}
    
    Now, using Eq.~\ref{eq:master}, we have that 
    \begin{eqnarray}
    \nonumber
    \frac{dG(t)}{dt} = \frac{1}{\langle x \rangle}\int_{0}^{\infty} dx \int_{0}^{\infty} dx' \int_{-x}^{x'} d\Delta\> \phi(x,t)\> \omega_{[x+\Delta,x'-\Delta]\rightarrow[x,x']}f(x+\Delta,t) f(x'-\Delta,t) -\\
    -\frac{1}{\langle x \rangle}\int_{0}^{\infty} dx \int_{0}^{\infty} dx' \int_{-x}^{x'} d\Delta\> \phi(x,t)\> \omega_{[x,x']\rightarrow[x+\Delta,x'-\Delta]}f(x,t) f(x',t).\>\>\>\>\>
    \label{eq:dgdt}
    \end{eqnarray}
    
    We can represent the first integral by
    \begin{eqnarray}
    \nonumber
    I = \int_{0}^{\infty} dx \int_{0}^{\infty} dx' \int_{-x}^{x'} d\Delta\> \phi(x,t)\> \omega_{[x+\Delta,x'-\Delta]\rightarrow[x,x']} \> f(x+\Delta,t) f(x'-\Delta,t)=\\
    \nonumber
    = \int_{-\infty}^{\infty} dx \int_{-\infty}^{\infty} dx' \int_{-\infty}^{\infty} d\Delta \> \phi(x,t) \>\omega_{[x+\Delta,x'-\Delta]\rightarrow[x,x']} \> f(x+\Delta,t) f(x'-\Delta,t)  \times\\
    \times \theta(x) \theta(x') \theta(\Delta + x)\theta(x'-\Delta).
    \end{eqnarray}
    
    Defining $y = x + \Delta$ e $y' = x'-\Delta$, we get
    \begin{eqnarray}
    \nonumber
    I = \int_{-\infty}^{\infty} dy \int_{-\infty}^{\infty} dy' \int_{-\infty}^{\infty} d\Delta \>\phi(y-\Delta,t)\>  
    \omega_{[y,y']\rightarrow[y-\Delta,y'+\Delta]} \> f(y,t) f(y',t)\times\\
    \times \theta(y-\Delta)\theta(y'+\Delta)\theta(y)\theta(y').
    \end{eqnarray}
    
    With a new variable substitution, $\Delta' = -\Delta$, the integral becomes
    \begin{eqnarray}
    \nonumber
    I = \int_{-\infty}^{\infty} dy \int_{-\infty}^{\infty} dy' \int_{-\infty}^{\infty} d\Delta' \>\phi(y+\Delta',t)\>  
    \omega_{[y,y']\rightarrow[y+\Delta',y'-\Delta']} \> f(y,t) f(y',t)\times\\
    \nonumber
    \times \theta(y+\Delta')\theta(y'-\Delta')\theta(y)\theta(y') = \\
    =\int_{0}^{\infty} dy \int_{0}^{\infty} dy' \int_{-y}^{y'} d\Delta' \>\phi(y+\Delta',t)\>  
    \omega_{[y,y']\rightarrow[y+\Delta',y'-\Delta']} \> f(y,t) f(y',t),
    \label{eq:inty}
    \end{eqnarray}
     in which the change in the differential signal was compensated by permuting the limits in the integral over $\Delta'$.
    
    Comparing the integral Eq.~\ref{eq:inty} with the second integral of Eq.\ref{eq:dgdt}, we see that the only difference between them, besides the name of variables, is the signal and the argument of $\phi$. Thus, the time derivative of Gini, Eq.~\ref{eq:dgdt}, becomes
    \begin{eqnarray}
    \nonumber
    \frac{dG(t)}{dt} = \frac{1}{\langle x \rangle}\int_{0}^{\infty} dx \int_{0}^{\infty} dx'f(x,t) f(x',t) \times\\
    \times\Bigg[\int_{-x}^{x'} d\Delta\> \omega_{[x,x']\rightarrow[x+\Delta,x'-\Delta]}\> \Big(\phi(x + \Delta,t) - \phi(x,t)\Big) \Bigg].
    \end{eqnarray}
        
    With the definition of $\phi$ (Eq.~\ref{eq:phi}) and using the normalization of the transfer rate (Eq.~\ref{eq:normal}), we finally get
   	\begin{eqnarray}\label{eq:gini_ev}
	\nonumber
	\frac{dG(t)}{dt} = \frac{1}{\langle x \rangle}\int_{0}^{\infty} dx \int_{0}^{\infty} dx'\int_{0}^{\infty} d x_1\> f(x,t) f(x',t) f(x_1,t)\times \\
	\times \>\bigg[\int_{-x}^{x'} d\Delta\> \omega_{[x,x']\rightarrow[x+\Delta,x'-\Delta]}\>|x+\Delta-x_1| - |x-x_1|\bigg]\>.
	\end{eqnarray}
	
	To prove the proposition, it is enough to show that the term in the brackets is always non-negative. Let us first note that
	\begin{equation}
	-|x + \Delta - x_1| \leq (x + \Delta - x_1) \leq |x + \Delta - x_1|.
	\end{equation}
	Multiplying the terms by $\omega_{[x,x']\rightarrow[x+\Delta,x'-\Delta]}$ and integrating it over $\Delta$ in the domain $[-x,x']$, we then have
	\begin{eqnarray} 
	\nonumber
	-\int_{-x}^{x'} d\Delta\> \omega_{[x,x']\rightarrow[x+\Delta,x'-\Delta]}|x + \Delta - x_1| \leq \\
	\nonumber
	\leq \int_{-x}^{x'} d\Delta\> \omega_{[x,x']\rightarrow[x+\Delta,x'-\Delta]}(x + \Delta - x_1) \leq \\
	\leq \int_{-x}^{x'} d\Delta\> \omega_{[x,x']\rightarrow[x+\Delta,x'-\Delta]}|x + \Delta - x_1|.
	\label{eq:mod}
	\end{eqnarray}
	
	Second, we can note that, due to the property of the unbiased transfer rate (Eq.~\ref{eq:unbiased}) and its normalization (Eq.~\ref{eq:normal}), the second term in the inequality  Eq.~\ref{eq:mod} is
	\begin{equation}
	\int_{-x}^{x'} d\Delta\> \omega_{[x,x']\rightarrow[x+\Delta,x'-\Delta]}(x + \Delta - x_1) = (x - x_1).
	\end{equation}
Therefore, Eq.~\ref{eq:mod} reduces to 
	\begin{eqnarray} 
	\nonumber
	-\int_{-x}^{x'} d\Delta\> \omega_{[x,x']\rightarrow[x+\Delta,x'-\Delta]}|x + \Delta - x_1| \leq\\
	\nonumber
	\leq (x - x_1) \leq \\
	\leq \int_{-x}^{x'} d\Delta\> \omega_{[x,x']\rightarrow[x+\Delta,x'-\Delta]}|x + \Delta - x_1|.
	\end{eqnarray}
	As the only difference between the inferior and superior limits of $(x - x_1)$ is a sign, we finally arrive to
	\begin{equation}
	|x - x_1| \leq \int_{-x}^{x'} d\Delta\> \omega_{[x,x']\rightarrow[x+\Delta,x'-\Delta]}|x + \Delta - x_1|,
	\end{equation}
	which  proves that
	\begin{eqnarray}
	\frac{dG(t)}{dt} \geq 0.
	\end{eqnarray}
\end{proof}

\begin{proof}[{\bf proof of proposition 3}]
	Since the Gini Index is limited above by $G(t) \leq 1$, it will increase until its stationary value $G_s$ associated with the density probability function $f_s(x)$, defined as
	\begin{eqnarray}
	G_s = \lim_{t\rightarrow\infty}G(t) \>\>\>\text{e}\>\>\> f_s(x) = \lim_{t\rightarrow\infty}f(x,t).
	\end{eqnarray}
	
	Also, for each value of $x$, we define its stationary mobility as
	\begin{eqnarray}
	l_s(x) = \lim_{t\rightarrow\infty}l(x,t).
	\end{eqnarray}
	Since the integrand of Eq.~\ref{eq:gini_ev} is always non-negative, the stationary state is only achieved when
	\begin{equation}
	\int_0^{\infty}dx'\Bigg[\int_{-x}^{x'} d\Delta\> \omega_{[x,x']\rightarrow[x+\Delta,x'-\Delta]}\>|x + \Delta - x_1| - |x - x_1|\Bigg]\>f_s(x) f_s(x_1) f_s(x') = 0
	\end{equation}
	for all $x,x_1 \in [0,\infty)$. For $x = 0$, due to the proposition 1, this equality is trivially satisfied. Now, assuming that $x=0$ is the unique absorbing state ($l_s(x) > 0$ for all $x > 0$), there is no $x > 0$ such that $f_s(x) > 0$, since the integrand in the particular case $x_1 = x$ would be
	\begin{equation}
	\int_0^{\infty}dx'\int_{-x}^{x'} d\Delta\> \omega_{[x,x']\rightarrow[x+\Delta,x'-\Delta]}\>|\Delta|f_s(x) f_s(x) f_s(x') = \Big(f_s(x)\Big)^2l_s(x) > 0.
	\end{equation}
	Then, the integrand is only zero if $f_s(x) = 0$ for all $x > 0$. The unique probability density function that satisfies this condition is the absolute oligarchy one
	\begin{equation}
	f_s(x) = \delta (x) + \lim_{N \rightarrow \infty} \frac{\delta (x - \langle x\rangle N)}{N},
	\end{equation}
	associated with the Gini Index $G_s = 1$.
	
	Finally, the stationary value of Liquidity (Eq.~\ref{eq:liquidez}) os given by
	\begin{eqnarray}
	L_s\equiv\lim_{t\rightarrow\infty}L(t) = \frac{1}{2\langle x \rangle}\int_{0}^{\infty} dx\>l_s(x)\>f_s(x) = \frac{l_s(0)}{2\langle x \rangle} + \frac{1}{2\langle x \rangle}\lim_{N\rightarrow \infty}\frac{1}{N}l_s\Big(N\langle x \rangle\Big).
	\end{eqnarray} 
	The first term is zero by the Proposition 1. Moreover, by the Corollary~\ref{cor}, we have that
	\begin{eqnarray}
	\nonumber
	0\leq\frac{1}{2\langle x \rangle}\lim_{N\rightarrow \infty}\frac{1}{N}l_s\Big(N\langle x \rangle\Big) \leq \frac{1}{2\langle x \rangle}\lim_{N\rightarrow \infty}\frac{2\langle x \rangle}{N} = \lim_{N\rightarrow \infty}\frac{1}{N} = 0 
	\Rightarrow \\
	\Rightarrow\frac{1}{2\langle x \rangle}\lim_{N\rightarrow \infty}\frac{1}{N}l_s\Big(N\langle x \rangle\Big) = 0.
	\end{eqnarray}
	
	This proves that $L_s = 0$.

\end{proof}

\section{Conclusions}

The second law of thermodynamics has as a corollary the thermal death of the universe. 
We have presented the proof that any model of binary wealth exchange based on {\em a priori} unbiased rule, {\em i.e.}, no bias toward any agent in every microscopic wealth exchange, will end up in the ``thermal death of the market'' at the maximum unequal state.

There are previous
numerical~\cite{hayes2002,cardoso2020wealth,bouleau2017impact,iglesias2012entropy} and analytical~\cite{bouleau2017impact,boghosian2015h,moukarzel2007wealth,moukarzel2011multiplicative} restricted demonstrations that some unbiased binary exchange rules lead to condensation.
The present one, however, is the first general analytical demonstration that condensation occurs for all unbiased games, regardless of its particularities.
		
Thus, what is the implications of the prepositions here demonstrated in the real world? The worldwide data about economies show a clear tendency toward increasing inequalities~\cite{piketty2018distributional}. As for an example, we can cite the US case, where the fraction of national wealth in hands of the top $1\%$ risen from $22\%$ in $1980$ to $37\%$ in $2015$~\cite{piketty2018distributional}. 

We might think that a system without regulatory policies is {\it fair} since no individual has systematic or {\it a priori} advantages~\cite{fargione2011entrepreneurs,biondi2019inequality,biondi2020financial,klass2007forbes,levy2003investment}. However, we have proved that such apparent ``fairness'' is the cause of the growing inequality.

The only restrictions of the present demonstration are pairwise exchanges and the conservation of the total wealth. 
It is our feeling that the demonstration could be extended to many-body interaction, as we are planning to explore.  The same could be said for a growing economy, with a reformulation of the impartial market.

\end{document}